\documentclass[11pt]{article}

\usepackage[utf8]{inputenc}
\usepackage[T1]{fontenc}
\usepackage[margin=1in]{geometry}
\usepackage{amsmath,amssymb,amsthm}
\usepackage{booktabs}
\usepackage{hyperref}
\usepackage{tikz}

\newtheorem{theorem}{Theorem}
\newtheorem{lemma}[theorem]{Lemma}
\newtheorem{corollary}[theorem]{Corollary}
\theoremstyle{definition}
\newtheorem{definition}[theorem]{Definition}

\newcommand{\Z}{\mathbb{Z}}
\newcommand{\N}{\mathbb{N}}
\newcommand{\Spir}{\Sigma}
\newcommand{\pieces}{\operatorname{pieces}}
\newcommand{\heavyP}{\operatorname{heavyP}}
\newcommand{\costly}{\operatorname{costly}}
\newcommand{\len}{\operatorname{len}}
\newcommand{\Band}{\operatorname{Band}}
\newcommand{\indeg}{\operatorname{indeg}}
\newcommand{\outdeg}{\operatorname{outdeg}}
\newcommand{\lean}[1]{\texttt{\detokenize{#1}}}

\title{Linear-Time FPT Rectilinear Planarity Testing\\ of Series-Parallel Graphs}
\author{Giacomo Ortali\\ Universit\`a Telematica Internazionale Uninettuno\\ \texttt{giacomo.ortali@uninettunouniversity.net}}
\date{}

\begin{document}
\maketitle

\begin{abstract}
A planar graph is rectilinear planar if it admits a planar orthogonal drawing without bends. For series-parallel graphs in the variable embedding setting the best known testing algorithm takes $O(n^2)$ time. We show that rectilinear planarity of a series-parallel graph with maximum degree four can be tested in $2^{O(k)}\cdot n$ time, where $k$ is the number of degree-4 vertices; the planar embedding is not fixed. In particular, the running time is subquadratic when $k=o(\log n)$. Each node of the SPQ$^*$-tree stores the set of its admissible spiralities; we show that this set is a union of $2^{O(k)}$ arithmetic progressions of step two. All proofs are in the paper and are additionally verified in Lean~4 with Mathlib.
\end{abstract}

\section{Introduction}\label{sec:intro}

A graph is \emph{rectilinear planar} if it admits a planar orthogonal drawing without bends; only graphs of maximum degree four (\emph{4-graphs}) can be rectilinear planar. Rectilinear planarity testing is NP-hard~\cite{GarTam01}. For partial 2-trees in the variable embedding setting it can be solved in $O(n^2)$ time, and in $O(n)$ time when no two parallel components of a biconnected component share a pole~\cite{DKLO22}. Whether $O(n)$ time is achievable for all series-parallel graphs is open.
\begin{theorem}\label{thm:main}
Let $G$ be a series-parallel 4-graph with $n$ vertices, $k$ having degree four. Whether $G$ is rectilinear planar, over all planar embeddings, can be tested in $2^{O(k)}\cdot n$ time.
\end{theorem}
In particular, if $k=o(\log n)$ the running time is $n^{1+o(1)}$, which is subquadratic and improves on the $O(n^2)$ bound of~\cite{DKLO22}.

We follow the spirality approach of~\cite{DKLO22}: each node $\mu$ of the SPQ$^*$-tree gets the set $\Spir(\mu)$ of spiralities its component can take. The $O(n^2)$ bound of~\cite{DKLO22} comes from these sets not being intervals. Without degree-4 vertices the spirality sets are intervals; degree-4 vertices may break them into several maximal arithmetic progressions of step two (\emph{pieces}). Our main technical result, the Key Lemma (Lemma~\ref{lem:key}), shows that a P2-node without degree-4 poles does not increase the number of pieces when the spirality set of one of its children has a single piece; as a consequence, every spirality set has at most $36^{4k}$ pieces (Lemma~\ref{lem:global}). Storing each set as the list of its pieces, every node is then processed in $2^{O(k)}$ time.
All proofs are given in the paper. The Key Lemma, the combinatorial bounds, the correctness of the algorithm and its cost for a fixed root are additionally verified in Lean~4 with Mathlib (Section~\ref{sec:lean}).

\paragraph{Outline.} Section~\ref{sec:prelim} recalls spiralities and the rules of~\cite{DKLO22}. Section~\ref{sec:pieces} bounds the number of pieces of every spirality set. Section~\ref{sec:algo} gives the algorithm, first for a fixed root and then for all roots. Section~\ref{sec:lean} describes the formalization.

\section{Preliminaries}\label{sec:prelim}

\paragraph{Series-parallel graphs and SPQ$^*$-trees.} For basic notions on series-parallel graphs, orthogonal representations and SPQ$^*$-trees we refer to~\cite{DETT99,DKLO22}. Series-parallel graphs are biconnected; graphs whose blocks are series-parallel are partial 2-trees. Throughout, $G$ is a series-parallel 4-graph with $n$ vertices, not a simple cycle (a cycle is rectilinear planar iff it has at least four vertices), and $T$ is its SPQ$^*$-tree, rooted at a Q$^*$-node. A Q$^*$-node $\mu$ is a chain with $\len(\mu)$ edges; a P-node has two or three children (P2- and P3-nodes); each node $\mu$ has two poles $u,v$ and a pertinent graph $G_\mu$. $T$ has $O(n)$ nodes and is computed in $O(n)$ time~\cite{GM00}. See Figure~\ref{fig:spqr} for an example.

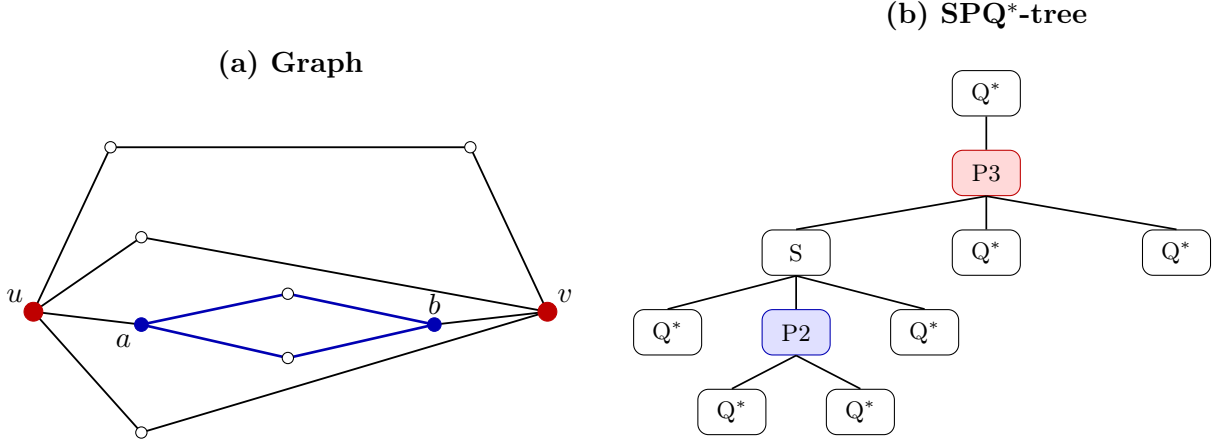
\begin{figure}[tbp]
\centering
\begin{tikzpicture}[
  x=.68cm,y=.68cm,
  edge/.style={draw=black,line width=.7pt},
  light edge/.style={draw=blue!70!black,line width=1pt},
  vertex/.style={circle,draw=black,fill=white,inner sep=1.5pt},
  light vertex/.style={circle,draw=blue!70!black,fill=blue!70!black,inner sep=1.8pt},
  pole/.style={circle,draw=red!75!black,fill=red!75!black,inner sep=2.5pt},
  tree edge/.style={draw=black,line width=.65pt},
  tree node/.style={draw,rounded corners,fill=white,minimum width=9mm,minimum height=6mm,inner sep=.8pt,font=\footnotesize},
  heavy p/.style={draw=red!75!black,fill=red!15},
  light p/.style={draw=blue!70!black,fill=blue!12}
]
  \node[font=\normalsize\bfseries] at (5,4.8) {(a) Graph};
  \coordinate (u) at (0,0);
  \coordinate (v) at (10,0);
  \coordinate (r1) at (1.5,3.2);
  \coordinate (r2) at (8.5,3.2);
  \coordinate (c) at (2.1,1.45);
  \coordinate (d) at (2.1,-2.35);
  \coordinate (a) at (2.1,-.25);
  \coordinate (b) at (7.8,-.25);
  \coordinate (x) at (4.95,.35);
  \coordinate (y) at (4.95,-.9);
  \draw[edge] (u)--(r1)--(r2)--(v);
  \draw[edge] (u)--(c)--(v);
  \draw[edge] (u)--(d)--(v);
  \draw[edge] (u)--(a);
  \draw[light edge] (a)--(x)--(b);
  \draw[light edge] (a)--(y)--(b);
  \draw[edge] (b)--(v);
  \foreach \z in {r1,r2,c,d,x,y} \node[vertex] at (\z) {};
  \node[pole] at (u) {};
  \node[pole] at (v) {};
  \node[light vertex] at (a) {};
  \node[light vertex] at (b) {};
  \node[above left,font=\normalsize] at (u) {$u$};
  \node[above right,font=\normalsize] at (v) {$v$};
  \node[below left,font=\normalsize] at (a) {$a$};
  \node[above=0.01cm,font=\normalsize] at (b) {$b$};

  \begin{scope}[xshift=1.25cm]
  \node[font=\normalsize\bfseries] at (16.7,5.75) {(b) SPQ$^*$-tree};
  \node[tree node] (qroot) at (16.7,4.25) {Q$^*$};
  \node[tree node,heavy p] (p3) at (16.7,2.7) {P3};
  \node[tree node] (snode) at (13,1.15) {S};
  \node[tree node] (qb) at (16.7,1.15) {Q$^*$};
  \node[tree node] (qd) at (20.4,1.15) {Q$^*$};
  \node[tree node] (qua) at (10.5,-.4) {Q$^*$};
  \node[tree node,light p] (p2) at (13, -.4) {P2};
  \node[tree node] (qbv) at (15.5,-.4) {Q$^*$};
  \node[tree node] (qax) at (11.75,-1.95) {Q$^*$};
  \node[tree node] (qay) at (14.25,-1.95) {Q$^*$};
  \draw[tree edge] (qroot.south)--(p3.north);
  \draw[tree edge] (p3.south)--(snode.north);
  \draw[tree edge] (p3.south)--(qb.north);
  \draw[tree edge] (p3.south)--(qd.north);
  \draw[tree edge] (snode.south)--(qua.north);
  \draw[tree edge] (snode.south)--(p2.north);
  \draw[tree edge] (snode.south)--(qbv.north);
  \draw[tree edge] (p2.south)--(qax.north);
  \draw[tree edge] (p2.south)--(qay.north);
  \end{scope}
\end{tikzpicture}
\caption{(a) A series-parallel 4-graph; $u$ and $v$ have degree four. (b) Its SPQ$^*$-tree, rooted at the top chain: the P3-node with poles $u,v$ is heavy (red), the P2-node with poles $a,b$ is light (blue).}
\label{fig:spqr}
\end{figure}

\paragraph{Spirality.} Let $H$ be a rectilinear representation of $G$ and $H_\mu$ its restriction to $G_\mu$. Intuitively, the spirality of $H_\mu$ measures how much $G_\mu$ is rolled up between its poles: it counts the right minus left turns along a path from one pole to the other, in units of $90^\circ$. Formally~\cite{DiBattistaLV98,DKLO22}, for a pole $w$ let $\indeg_\mu(w)$ and $\outdeg_\mu(w)$ be the numbers of edges at $w$ inside and outside $G_\mu$. The \emph{alias vertices} of $w$ are: $w$ itself if $\indeg_\mu(w)=1$; two dummy vertices subdividing the two edges at $w$ outside $G_\mu$ if $\indeg_\mu(w)=\outdeg_\mu(w)=2$; one dummy vertex subdividing the edge at $w$ outside $G_\mu$ otherwise. A \emph{spine} is a path from an alias vertex $u'$ of $u$ to an alias vertex $v'$ of $v$ formed by $(u',u)$, a simple $u$-$v$ path in $G_\mu$, and $(v,v')$; let $n(S)$ be the number of right minus left turns along a spine $S$. The \emph{spirality} $\sigma_H(\mu)$ is $n(S)$ if each pole has one alias vertex, and otherwise the average of $n(S)$ over the two spines through the two alias vertices of a pole (if both poles have two alias vertices, over the two spines on the same side). It does not depend on the chosen path. A pole with two alias vertices contributes the average of two turns, so $\sigma_H(\mu)$ is half-integral when exactly one pole has two alias vertices, and integral otherwise. Figure~\ref{fig:spirality} shows one example of each case. 

\begin{figure}[tbp]
\centering
\begin{tikzpicture}[
  x=.68cm,y=.68cm,
  spine one/.style={draw=black,line width=1.6pt},
  spine two/.style={draw=black,dash pattern=on 3pt off 1.5pt,line width=1.35pt},
  alias edge/.style={draw=black,dash pattern=on 1pt off 1.2pt,line width=1.1pt},
  alias mask/.style={draw=white,line width=2.5pt},
  vertex/.style={circle,draw=black,fill=black,inner sep=1.5pt},
  pole/.style={circle,draw=black,fill=black,inner sep=2.2pt},
  degree four/.style={circle,draw=red!75!black,fill=red!75!black,inner sep=2.5pt},
  alias/.style={rectangle,draw=black,fill=white,minimum size=4.2pt,inner sep=0pt}
]
  \node[font=\normalsize\bfseries] at (1.5,3.35) {(a) One alias per pole};
  \coordinate (ua) at (-1,0);
  \coordinate (u) at (0,0);
  \coordinate (p) at (3,0);
  \coordinate (v) at (3,-2);
  \coordinate (q) at (0,2);
  \coordinate (r) at (4,2);
  \coordinate (t) at (4,-2);
  \coordinate (va) at (3,-3);
  \draw[spine two] (u)--(q)--(r)--(t)--(v);
  \draw[spine one] (ua)--(u)--(p)--(v)--(va);
  \draw[alias mask] (ua)--(u);
  \draw[alias mask] (v)--(va);
  \draw[alias edge] (ua)--(u);
  \draw[alias edge] (v)--(va);
  \foreach \z in {p,q,r,t} \node[vertex] at (\z) {};
  \node[pole] at (u) {};
  \node[pole] at (v) {};
  \node[alias] at (ua) {};
  \node[alias] at (va) {};
  \node[above left,font=\normalsize] at (u) {$u$};
  \node[below right,font=\normalsize] at (v) {$v$};
  \node[left,font=\normalsize] at (ua) {$u'$};
  \node[right,font=\normalsize] at (va) {$v'$};
  \node[font=\normalsize] at (1.5,-3.8) {$n(S)=1,\quad \sigma_H(\mu)=1$};

  \node[font=\normalsize\bfseries] at (13.5,3.35) {(b) Two aliases at one pole};
  \coordinate (u2) at (10,0);
  \coordinate (v2) at (16,-2);
  \coordinate (p2) at (16,0);
  \coordinate (q2) at (10,2);
  \coordinate (r2) at (18,2);
  \coordinate (t2) at (18,-2);
  \coordinate (u2a) at (9,0);
  \coordinate (u2b) at (10,-1);
  \coordinate (v2a) at (16,-3);
  \draw[spine one] (u2)--(p2)--(v2);
  \draw[spine two] (u2)--(q2)--(r2)--(t2)--(v2);
  \draw[alias mask] (u2a)--(u2);
  \draw[alias mask] (u2b)--(u2);
  \draw[alias mask] (v2)--(v2a);
  \draw[alias edge] (u2a)--(u2);
  \draw[alias edge] (u2b)--(u2);
  \draw[alias edge] (v2)--(v2a);
  \foreach \z in {p2,q2,r2,t2} \node[vertex] at (\z) {};
  \node[degree four] at (u2) {};
  \node[pole] at (v2) {};
  \node[alias] at (u2a) {};
  \node[alias] at (u2b) {};
  \node[alias] at (v2a) {};
  \node[above left,font=\normalsize,text=red!75!black] at (u2) {$u$};
  \node[below left,font=\normalsize] at (v2) {$v$};
  \node[left,font=\normalsize] at (u2a) {$u'_1$};
  \node[right,font=\normalsize] at (u2b) {$u'_2$};
  \node[right,font=\normalsize] at (v2a) {$v'$};
  \node[font=\normalsize] at (13.5,-3.8) {$n(S_1)=1,\quad n(S_2)=2,\quad \sigma_H(\mu)=\frac{1+2}{2}=\frac32$};
\end{tikzpicture}
\caption{A component $G_\mu$ with poles $u,v$ (solid and dashed edges); dotted edges lie outside $G_\mu$ and end at the alias vertices (squares); spines are drawn thick. (a) Each pole has two edges inside and one outside $G_\mu$, hence one alias vertex: $\sigma_H(\mu)=n(S)=1$. (b) $u$ has two edges inside and two outside $G_\mu$, hence two alias vertices: $\sigma_H(\mu)=(n(S_1)+n(S_2))/2=3/2$.}
\label{fig:spirality}
\end{figure}

\paragraph{Spirality sets.} We use \emph{doubled} spiralities, which are integers, and let $\Spir(\mu)\subseteq\Z$ be the set of doubled spiralities $2\sigma_H(\mu)$ over all rectilinear representations $H$ of $G$. Mirroring a representation exchanges right and left turns; since the embedding is not fixed, every $\Spir(\mu)$ is \emph{symmetric}, i.e., $\Spir(\mu)=-\Spir(\mu)$. Clearly, $\Spir(\mu)$ is finite. For $A,B\subseteq\Z$ let $A\oplus B=\{a+b : a\in A, b\in B\}$, $A+d=\{a+d : a\in A\}$, and for $c\in\N$ let $\Band(c)=\{2s : s\in\Z,\ |s|\le c\}$. For example, $\Band(2)=\{-4,-2,0,2,4\}$. The following is~\cite[Lemmas~1--3]{DKLO22} in the doubled scale.

\begin{definition}\label{def:sigma}
\begin{description}
\item[Q$^*$-node:] $\Spir(\mu)=\Band(\len(\mu)-1)$.
\item[S-node] with children $\mu_1,\dots,\mu_h$: $\Spir(\mu)=\Spir(\mu_1)\oplus\dots\oplus\Spir(\mu_h)$.
\item[P3-node:] $z\in\Spir(\mu)$ iff for some ordering $(a,b,c)$ of the children there are $x\in\Spir(a)$, $y\in\Spir(b)$, $t\in\Spir(c)$ with $z=x-4=y=t+4$.
\item[P2-node] with children $\mu_l,\mu_r$: $z\in\Spir(\mu)$ iff for one of the two orders of the children there are $x\in\Spir(\mu_l)$, $y\in\Spir(\mu_r)$ and $\alpha^l_u,\alpha^r_u,\alpha^l_v,\alpha^r_v\in\{0,1\}$ with $1\le\alpha^l_w+\alpha^r_w\le2$ for $w\in\{u,v\}$ and
\[
z = x - 2k^l_u\alpha^l_u - 2k^l_v\alpha^l_v = y + 2k^r_u\alpha^r_u + 2k^r_v\alpha^r_v,
\]
where $\alpha^l_w=\alpha^r_w=1$ whenever $k^l_w=\tfrac12$ or $k^r_w=\tfrac12$.
\end{description}
\end{definition}

Here $k^d_w=1$ if $\indeg_{\mu_d}(w)=\outdeg_\mu(w)=1$ and $k^d_w=\tfrac12$ otherwise (text before~\cite[Lemma~3]{DKLO22}). A pole of degree three has $\indeg_\mu(w)=2$, $\outdeg_\mu(w)=1$, hence coefficients $1$; a pole of degree four has a coefficient $\tfrac12$, and there $\alpha^l_w=\alpha^r_w=1$ (paragraph before~\cite[Lemma~3]{DKLO22plane}). So all coefficients of a P2-node equal $1$ iff it has no pole of degree four. Both poles of a P3-node have degree four.

\paragraph{Test at the root.} Let the root $\rho$ be a chain with $\ell$ edges and let $\nu$ be its child.

\begin{lemma}\label{lem:root}
$G$ has a rectilinear representation with $\rho$ on the external face iff $8\in\Spir(\nu)\oplus\Band(\ell-1)$.
\end{lemma}
\begin{proof}
By~\cite[Lemma~6]{DDLMO23} the condition is $\sigma(H_\nu)-n(H^{uv}_\rho)=4$, where $n(H^{uv}_\rho)$ is the number of right minus left turns along $\rho$. The $\ell-1$ internal vertices of $\rho$ turn by $-1$, $0$ or $1$ each, so $n(H^{uv}_\rho)$ ranges over $[-(\ell-1),\ell-1]$; double and use the symmetry of $\Band$.
\end{proof}

\paragraph{Pieces.} Throughout, intervals are intervals of integers: $[a,b]=\{x\in\Z : a\le x\le b\}$. A \emph{run} is a set $\{a,a+2,\dots,b\}$ with $b-a$ even. For finite $S\subseteq\Z$, a \emph{piece} of $S$ is a maximal run contained in $S$, and $\pieces_2(S)=|\{x\in S : x-2\notin S\}|$ is the number of pieces of $S$. For example, $S=\{-4,-2,0,4,6,9\}$ has the three pieces $\{-4,-2,0\}$, $\{4,6\}$ and $\{9\}$, so $\pieces_2(S)=3$. Similarly, $\pieces(S)=|\{x\in S : x-1\notin S\}|$ is the number of maximal intervals contained in $S$; for example, $\pieces(\{-2,-1,0,3,4\})=2$. In particular, $\pieces(S)=0$ iff $S=\emptyset$, and $\pieces(S)=1$ iff $S$ is a nonempty interval. If all elements of $S$ are even, then $\pieces_2(S)=\pieces(S/2)$.

\section{Bounding the Number of Pieces}\label{sec:pieces}

\paragraph{Light and heavy P-nodes.} A P2-node is \emph{light} if all its coefficients are $1$, i.e., it has no pole of degree four; all other P-nodes are \emph{heavy}. In Figure~\ref{fig:spqr}, the P3-node with poles $u,v$ is heavy and the P2-node with poles $a,b$ is light. A subtree is \emph{light} if it has no heavy P-node, and $\heavyP(\mu)$ is the number of heavy P-nodes in the subtree of $\mu$. Heavy P-nodes are the only places where the number of pieces can grow by more than a constant; the next lemma bounds how many there are.

\begin{lemma}\label{lem:heavy}
$\heavyP(T)\le 2k$, for every choice of the root.
\end{lemma}
\begin{proof}
Every heavy P-node has a pole of degree four, and a vertex $w$ of degree four is a pole of at most two P-nodes. Indeed, a P-node with pole $w$ contains at least two edges of $w$, and at least one edge of $w$ lies outside it, since the root is a chain. Two P-nodes with pole $w$ either contain disjoint sets of edges of $w$, or one contains the other, and then it contains strictly more edges of $w$, since it has a further child incident to $w$. Hence three P-nodes with pole $w$ would give $w$ degree at least five.
\end{proof}

The Key Lemma (below) is a statement about the maximal intervals of sets of integers. To apply it to a light P2-node we halve the doubled spiralities of its children, which requires them to be even.

\begin{lemma}\label{lem:parity}
If $\nu$ is a light P2-node with children $\mu_l,\mu_r$, all elements of $\Spir(\mu_l)$ and $\Spir(\mu_r)$ are even.
\end{lemma}
\begin{proof}
Let $u,v$ be the poles of $\nu$. Since $\nu$ is light, all its coefficients $k^d_w$ equal $1$, that is, $\indeg_{\mu_d}(w)=\outdeg_\nu(w)=1$ for each child $\mu_d$ and each pole $w\in\{u,v\}$. Hence, in $G_{\mu_d}$, each pole $w$ has a single alias vertex, namely $w$ itself, and by the definition of spirality every spirality of $\mu_d$ is an integer.
\end{proof}

The Key Lemma also requires one of the two children to have an interval as spirality set. This holds whenever the subtree of that child is light; P3-nodes need not be considered, since they are always heavy.

\begin{lemma}\label{lem:light}
If the subtree of $\mu$ is light, then $\Spir(\mu)=\emptyset$ or $\Spir(\mu)=\Band(c)$ for some $c\in\N$.
\end{lemma}
\begin{proof}
By induction on the subtree of $\mu$, which contains no P3-node, since P3-nodes are heavy. For Q$^*$-nodes the claim is Definition~\ref{def:sigma}, and for S-nodes it follows from $\Band(c_1)\oplus\Band(c_2)=\Band(c_1+c_2)$. For a light P2-node whose children have spirality sets $\Band(m)$ and $\Band(n)$, \cite[Table~1]{DKLO22plane}, taking the union over the two orders of the children, gives $\Spir(\mu)=\Band(c)$ with $c=\min\{\max(m,n),m+2,n+2\}$ if $m+n\ge2$, and $\Spir(\mu)=\emptyset$ otherwise.
\end{proof}

\paragraph{The Key Lemma.} The Key Lemma is a statement about finite sets of integers. Let $B_c=[-c,c]$.
For $A,B\subseteq\Z$, let $R(A,B)$ be the set of values $z=x-s_l=y+s_r$ with $x\in A$, $y\in B$, $s_l=\alpha^l_u+\alpha^l_v$, $s_r=\alpha^r_u+\alpha^r_v$ and angle variables as in Definition~\ref{def:sigma}, and let $\sigma_P(A,B)=R(A,B)\cup R(B,A)$. Thus $2\,\sigma_P(A,B)$ is the spirality set of a light P2-node whose children have spirality sets $2A$ and $2B$, and $R(A,B)$ corresponds to one fixed order of the children. See Figure~\ref{fig:keylemma} for an example. The following lemma describes $R$ on intervals; it is the only place where the angle variables enter the proof of the Key Lemma.

\begin{lemma}\label{lem:interval}
Let $I$ be an interval. Then $z\in R(I,B_c)$ if and only if some $a\in I$ satisfies
\[
2-c\le a,\qquad a-2\le z\le a,\qquad z\le c+2.
\]
Consequently, $R(I,B_c)$ is an interval, and if $R(B_d,B_c)\neq\emptyset$, it contains $0$.
\end{lemma}
\begin{proof}
The characterization is a direct check on the angle variables. Thus each $a\in I$ yields an interval contained in $[a-2,a]$, and consecutive values of $a$ yield overlapping intervals, so their union is an interval. If $R(B_d,B_c)\neq\emptyset$, some $a\le d$ satisfies $2-c\le a$; then $a=\min\{d,2\}$ satisfies the conditions above with $z=0$, so $0\in R(B_d,B_c)$.
\end{proof}

\begin{lemma}[Key Lemma]\label{lem:key}
If $A$ is finite and symmetric, then $\pieces(\sigma_P(A,B_c))\le\pieces(A)$.
\end{lemma}
\begin{proof}
The idea is to remove the two outermost maximal intervals of $A$, which are mirror images of each other, show that together they contribute at most two pieces to $\sigma_P(A,B_c)$, and handle the rest by induction on $\pieces(A)$. We use two simple observations. 
\begin{itemize} 
\item[(i)]If $A$ is symmetric, mirroring exchanges the two orders of the children: $R(B_c,A)=-R(A,B_c)$ and $\sigma_P(A,B_c)=R(A,B_c)\cup-R(A,B_c)$.
\item[(ii)] Every value of $R(X,B_c)$ comes from a single element of $X$, so $R(\cdot,B_c)$ distributes over union; moreover, $\pieces(X\cup Y)\le\pieces(X)+\pieces(Y)$.
\end{itemize} 

\emph{Base cases.} If $\pieces(A)=0$, then $A=\emptyset$ and $\sigma_P(A,B_c)=\emptyset$. If $\pieces(A)=1$, $A=[-m,m]$ with $m=\max A$. By Lemma~\ref{lem:interval}, $R(A,B_c)$ is empty or an interval containing $0$, and so is its mirror image; by (i), $\sigma_P(A,B_c)$ is their union, which has at most one piece.

\emph{Inductive step.} Let $\pieces(A)\ge2$ and assume the claim for every finite symmetric set with fewer pieces. Let  $[l,m]$ be the maximal interval contained in $A$ that contains $m=\max A$; by symmetry, $[-m,-l]$ is the maximal interval contained in $A$ that contains $\min A=-m$. These two intervals are disjoint, as otherwise $A=[-m,m]$ would have one piece. It immediately follows that $l\ge1$. Since $l-1\notin A$ by maximality and $-l+1\notin A$ by symmetry, the two outermost intervals are separated from the rest $A'=A\cap[-l+1,l-1]$, which is symmetric. Hence $A=A'\cup[l,m]\cup[-m,-l]$ and $\pieces(A)=\pieces(A')+2$. By (i) and (ii),
\[
\sigma_P(A,B_c)=\sigma_P(A',B_c)\cup H\cup(-H),\qquad H=R([l,m],B_c)\cup-R([-m,-l],B_c),
\]
where $H$ and $-H$ collect the contribution of the two outermost intervals. We claim that $H$ is an interval. Both terms are intervals by Lemma~\ref{lem:interval}. If the second term is nonempty, some $a\le-l$ satisfies $2-c\le a$ in the characterization of Lemma~\ref{lem:interval}, so $c\ge l+2$; then the characterization holds for $a=z=l$ and for $a=z=-l$, so both terms contain $l$. Hence $H$ and $-H$ contribute at most two pieces, and since $\pieces(A')=\pieces(A)-2$, the inductive hypothesis applied to $A'$ and (ii) give $\pieces(\sigma_P(A,B_c))\le\pieces(A')+2=\pieces(A)$.
\end{proof}

\begin{corollary}\label{cor:lightP2}
Let $\mu$ be a light P2-node whose child $\mu_r$ roots a light subtree. Then $\pieces_2(\Spir(\mu))\le\pieces_2(\Spir(\mu_l))$.
\end{corollary}
\begin{proof}
The corollary rewrites the Key Lemma in terms of spirality sets. By Lemma~\ref{lem:light}, $\Spir(\mu_r)$ is either empty, and then so is $\Spir(\mu)$, or equal to $2B_c$. By Lemma~\ref{lem:parity}, $\Spir(\mu_l)=2A$ for a symmetric $A\subseteq\Z$. Hence $\Spir(\mu)=2\,\sigma_P(A,B_c)$. Doubling a set turns maximal intervals into pieces, i.e., $\pieces_2(2X)=\pieces(X)$, so the claim is exactly Lemma~\ref{lem:key}.
\end{proof}
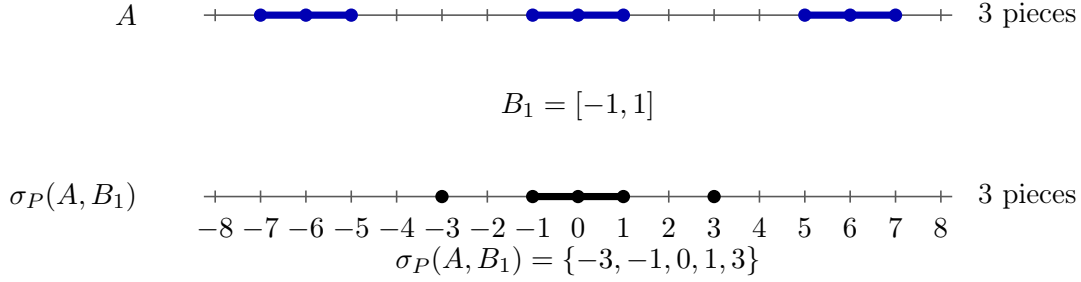
\begin{figure}[tbp]
\centering
\begin{tikzpicture}[
  x=.6cm,y=1cm,
  axis/.style={draw=black!65,line width=.6pt},
  input run/.style={draw=blue!70!black,line width=2.7pt},
  output run/.style={draw=black,line width=2.7pt}
]
  \node[anchor=east,font=\normalsize] at (-9.5,1.2) {$A$};
  \node[anchor=east,font=\normalsize] at (-9.5,-1.2) {$\sigma_P(A,B_1)$};
  \foreach \y in {1.2,-1.2} {
    \draw[axis] (-8.25,\y)--(8.25,\y);
    \foreach \x in {-8,...,8} \draw[axis] (\x,\y-.08)--(\x,\y+.08);
  }
  \draw[input run] (-7,1.2)--(-5,1.2);
  \draw[input run] (-1,1.2)--(1,1.2);
  \draw[input run] (5,1.2)--(7,1.2);
  \foreach \x in {-7,-6,-5,-1,0,1,5,6,7} \fill[blue!70!black] (\x,1.2) circle (2.5pt);
  \draw[output run] (-1,-1.2)--(1,-1.2);
  \foreach \x in {-3,-1,0,1,3} \fill[black] (\x,-1.2) circle (2.5pt);
  \foreach \x in {-8,...,8} \node[below=2pt,font=\normalsize] at (\x,-1.28) {$\x$};
  \node[font=\normalsize] at (0,0) {$B_1=[-1,1]$};
  \node[anchor=west,font=\normalsize] at (8.6,1.2) {3 pieces};
  \node[anchor=west,font=\normalsize] at (8.6,-1.2) {3 pieces};
  \node[font=\normalsize] at (0,-2.05) {$\sigma_P(A,B_1)=\{-3,-1,0,1,3\}$};
\end{tikzpicture}
\caption{A symmetric set $A$ with three pieces (top) and $\sigma_P(A,B_1)=\{-3,-1,0,1,3\}$ (bottom), which also has three pieces, as guaranteed by the Key Lemma.}
\label{fig:keylemma}
\end{figure}

\paragraph{Other compositions.} For finite $X\subseteq\Z$ write $p_X=\pieces_2(X)$. Outside the case covered by Corollary~\ref{cor:lightP2}, we only need a crude bound on how much each rule can increase the number of pieces.

\begin{lemma}\label{lem:compose}
Let $A,B,C\subseteq\Z$ be finite. Then $\pieces_2(A\oplus B)\le p_Ap_B$. Moreover, for any coefficients, the P2-rule applied to $A,B$ yields at most $18(p_A+p_B)$ pieces, and the P3-rule applied to $A,B,C$ yields at most $6(p_A+p_B+p_C)$ pieces.
\end{lemma}
\begin{proof}
We use two simple facts: (i) a shift $X+t$ has the same number of pieces as $X$; (ii) $\pieces_2(X\cup Y)\le p_X+p_Y$ and $\pieces_2(X\cap Y)\le p_X+p_Y$, since a piece of $X\cup Y$ or of $X\cap Y$ starts at some $x$ with $x-2$ missing from $X$ or from $Y$, so $x$ is the first element of a piece of $X$ or of $Y$.

\emph{S-rule.} $A\oplus B$ is the union of the sums of a piece of $A$ and a piece of $B$. The sum of $\{a,a+2,\dots,b\}$ and $\{c,c+2,\dots,d\}$ is $\{a+c,a+c+2,\dots,b+d\}$, so by (ii) $A\oplus B$ has at most $p_Ap_B$ pieces.

\emph{P2-rule.} Once the order of the two children and the values of the angle variables are fixed, the rule constrains the spirality of each child to be a fixed shift of the spirality of $\mu$; hence the admissible values form a set $(A-s_1)\cap(B+s_2)$, which by (i) and (ii) has at most $p_A+p_B$ pieces. There are $9$ admissible angle assignments and $2$ orders, so by (ii) the P2-rule yields at most $18(p_A+p_B)$ pieces.

\emph{P3-rule.} For each of the $6$ orderings of the children, the rule yields $(A-4)\cap B\cap(C+4)$, which by (i) and (ii) has at most $p_A+p_B+p_C$ pieces; by (ii), the union over the orderings has at most $6(p_A+p_B+p_C)$ pieces.
\end{proof}

\paragraph{Global bound.} We now combine the previous lemmas along the SPQ$^*$-tree. A light subtree contributes at most one piece, and the number of pieces can grow only at the following nodes, which we call \emph{costly}: heavy P-nodes, light P2-nodes whose two children both root non-light subtrees, and S-nodes with at least two children rooting non-light subtrees. Let $\costly(\mu)$ be the number of costly nodes in the subtree of $\mu$.

\begin{lemma}\label{lem:global}
For every node $\mu$, $\pieces_2(\Spir(\mu))\le 36^{\costly(\mu)}\le 36^{2\heavyP(\mu)}\le 36^{4k}$.
\end{lemma}
\begin{proof}
\emph{First inequality.} We proceed by induction on the subtree of $\mu$. If the subtree of $\mu$ is light, then $\costly(\mu)=0$ and, by Lemma~\ref{lem:light}, $\Spir(\mu)$ is empty or $\Band(c)=2B_c$, which has one piece. Otherwise, for each child $\mu_i$ of $\mu$ write $x_i=36^{\costly(\mu_i)}\ge1$; by induction, $\pieces_2(\Spir(\mu_i))\le x_i$.
\begin{itemize}
\item If $\mu$ is an S-node, applying Lemma~\ref{lem:compose} repeatedly gives $\pieces_2(\Spir(\mu))\le\prod_i x_i=36^{\sum_i\costly(\mu_i)}\le36^{\costly(\mu)}$.
\item If $\mu$ is a light P2-node with a child, say $\mu_r$, rooting a light subtree, then $\mu$ is not costly and $\costly(\mu_r)=0$; by Corollary~\ref{cor:lightP2}, $\pieces_2(\Spir(\mu))\le x_l=36^{\costly(\mu)}$.
\item Otherwise $\mu$ is a costly P-node, and Lemma~\ref{lem:compose} gives at most $18(x_1+x_2)\le36x_1x_2$ pieces for a P2-node and at most $6(x_1+x_2+x_3)\le36x_1x_2x_3$ for a P3-node, using $x_i\ge1$. In both cases the bound is $36^{\costly(\mu)}$.
\end{itemize}

\emph{Second inequality.} We show $\costly(\mu)\le2\heavyP(\mu)-1$ whenever the subtree of $\mu$ is not light; for light subtrees both counts are $0$. Write $c_i=\costly(\mu_i)$ and $h_i=\heavyP(\mu_i)$; by induction, $c_i\le2h_i-1$ for the children rooting non-light subtrees and $c_i=h_i=0$ for the others, so $c_i\le 2h_i$ for all children.
\begin{itemize}
\item At a heavy P-node, $\heavyP(\mu)=\sum_ih_i+1$, so $\costly(\mu)=\sum_ic_i+1\le2\sum_ih_i+1=2\heavyP(\mu)-1$.
\item At a light P2-node with a child rooting a light subtree, both counts equal those of the other child.
\item At a light P2-node with two non-light children, $\costly(\mu)=c_l+c_r+1\le(2h_l-1)+(2h_r-1)+1=2\heavyP(\mu)-1$.
\item At an S-node with $j\ge1$ non-light children, the node itself counts only if $j\ge2$, so, summing over the non-light children only, $\costly(\mu)\le\sum_i(2h_i-1)+1=2\heavyP(\mu)-j+1\le2\heavyP(\mu)-1$ when $j\ge2$, and $\costly(\mu)\le2\heavyP(\mu)-1$ directly when $j=1$.
\end{itemize}

\emph{Third inequality.} It follows from $\heavyP(\mu)\le2k$ (Lemma~\ref{lem:heavy}).
\end{proof}

\section{The Algorithm}\label{sec:algo}

\paragraph{Representation.} A finite $S\subseteq\Z$ is stored as the sorted list of its pieces, each piece $\{a,a+2,\dots,b\}$ being given by its endpoints $a$ and $b$; the list has length $\pieces_2(S)$. A shift of a set of this form, and the sum of two such sets, are again of this form, and the intersection of two such sets is of this form or empty; so every operation first computes the result as a union of such sets and then merges them into the list of pieces. The rules of Definition~\ref{def:sigma} are evaluated as in the proof of Lemma~\ref{lem:compose}: S as iterated sums, P2 as the union over $2$ orders and $9$ admissible angle assignments of intersections of shifted lists, and P3 as the union over $6$ orderings $(\mu_a,\mu_b,\mu_c)$ of $(\Spir(\mu_a)-4)\cap\Spir(\mu_b)\cap(\Spir(\mu_c)+4)$.

\begin{lemma}\label{lem:correct}
For every node $\mu$, the computed list is the list of pieces of $\Spir(\mu)$, and its length is at most $36^{2\heavyP(\mu)}$.
\end{lemma}
\begin{proof}
Each operation is correct by the description above, and the length bound is Lemma~\ref{lem:global}.
\end{proof}

\paragraph{Running time for a fixed root.} Operations on integers have unit cost. Let $N=36^{4k}$; by Lemma~\ref{lem:global}, every set $\Spir(\mu)$ has at most $N$ pieces, so its list has at most $N$ elements. The rule of a P-node is a constant number of shifts, unions and intersections. The rule of an S-node $\mu$ with $d$ children $\mu_1,\dots,\mu_d$ computes $d-1$ sums; each partial sum $\Spir(\mu_1)\oplus\dots\oplus\Spir(\mu_j)$ has at most $\prod_{i\le j}36^{\costly(\mu_i)}\le36^{\costly(\mu)}\le N$ pieces, as in the S-node case of the proof of Lemma~\ref{lem:global}. All intermediate lists have length polynomial in $N$, and each operation, including the final merge, takes time polynomial in the lengths of its inputs. Hence a node with $d$ children costs $O(d)\cdot N^{O(1)}$.

\begin{lemma}\label{lem:time}
For a fixed root of $T$, all lists are computed in $2^{O(k)}\cdot n$ time.
\end{lemma}
\begin{proof}
Summing $O(d)\cdot N^{O(1)}$ over the nodes of $T$, whose numbers $d$ of children add up to $O(n)$, gives $O(n)\cdot N^{O(1)}=2^{O(k)}\cdot n$.
\end{proof}

\paragraph{Test at the root.} By Lemma~\ref{lem:root}, the test for the root chain $\rho$ computes $\Spir(\nu)\oplus\Band(\ell-1)$ and checks whether it contains $8$, in $2^{O(k)}$ further time.

\paragraph{Variable root and external face.} For a node $r$ of $T$, let $\Spir_r(\mu)$ denote the spirality set of $\mu$ when $T$ is rooted at $r$. Testing every root separately costs $2^{O(k)}\cdot n^2$. To avoid this we use the reusability principle of~\cite{DLOP20}, in the form of~\cite[Lemma~5.11]{DBLP:journals/corr/abs-2502-03309}, which we state for our setting.

\begin{lemma}\label{lem:reroot}
Let $r_1,\dots,r_h$ be nodes of $T$. Assume that, for every node $\mu$ with $d$ children in the rooting at $r_i$, the set $\Spir_{r_i}(\mu)$ is computed from the sets of its children with $O(d)$ operations if $i=1$, and with $O(1)$ operations if $i\ge2$. Then $\Spir_{r_1}(r_1),\dots,\Spir_{r_h}(r_h)$ are computed with $O(n)$ operations.
\end{lemma}
We apply it taking the operations on lists as elementary operations; since each of them costs $2^{O(k)}$, as shown below, this gives $2^{O(k)}\cdot n$ time.

A P-node has at most four neighbours in $T$, since its poles have degree at most four; hence its set is recomputed with $O(1)$ operations, and its coefficients, which depend only on the degrees of its poles, in $O(1)$ time. For an S-node with children $\nu_1,\dots,\nu_d$ in the first rooting, the first traversal also stores the prefix sums $\Spir(\nu_1)\oplus\dots\oplus\Spir(\nu_j)$ and the suffix sums $\Spir(\nu_j)\oplus\dots\oplus\Spir(\nu_d)$, with $O(d)$ operations. When $\nu_p$ becomes the parent, the new set is the sum of the prefix up to $\nu_{p-1}$, the suffix from $\nu_{p+1}$, and the set on the side of the old parent: $O(1)$ operations. Each prefix or suffix sum is a partial sum of an S-node, so it has at most $N$ pieces as in the running-time analysis, and the other sets are spirality sets $\Spir_{r_i}(\mu)$, bounded by Lemma~\ref{lem:global} in the rooting at $r_i$. Hence each list has length at most $N$, and each operation costs $2^{O(k)}$.

\begin{proof}[Proof of Theorem~\ref{thm:main}]
Compute $T$ in $O(n)$ time~\cite{GM00}. By Lemmas~\ref{lem:time} and~\ref{lem:reroot}, the sets required by Lemma~\ref{lem:root} for all root chains are computed in $2^{O(k)}\cdot n$ time. Every external face contains a chain, so all choices of the external face are covered.
\end{proof}

\section{Formal Verification}\label{sec:lean}

The formalization is a single Lean~4.35.0 file with Mathlib v4.35.0, available at
\begin{center}\url{https://github.com/GiacomOrtali/rectilinear-sp-linearfpt-lean}.\end{center}
SPQ$^*$-trees are an inductive type (Q$^*$-nodes with their length, S-nodes with a list of children, P2-nodes with four integer coefficients $2k^d_w$, P3-nodes), and $\Spir$ is defined by the rules of Definition~\ref{def:sigma}. The file contains no \lean{sorry}, and \texttt{\#print axioms} reports only \lean{propext}, \lean{Classical.choice}, \lean{Quot.sound}.

\begin{table}[h]
\centering
\begin{tabular}{@{}ll@{}}
\toprule
Statement & Lean \\
\midrule
Lemma~\ref{lem:key} (Key Lemma) & \lean{KeyLemma.key_lemma} \\
Symmetry and finiteness of $\Spir$ & \lean{T1_1}, \lean{T2_0} \\
Lemma~\ref{lem:light} & \lean{T1_2}, \lean{T1_3} \\
Corollary~\ref{cor:lightP2} & \lean{T2_1} \\
Lemma~\ref{lem:compose} & \lean{T3_1}--\lean{T3_4} \\
Lemma~\ref{lem:global} & \lean{T4_0}--\lean{T4_3}, \lean{T5_0}--\lean{T5_2} \\
Lemma~\ref{lem:correct} & \lean{T6a_1}--\lean{T6a_3}, \lean{T6a'_1}--\lean{T6a'_4} \\
Lemma~\ref{lem:time} & \lean{T6b} \\
Root test & \lean{T6d_1}, \lean{T6d_2} \\
\bottomrule
\end{tabular}
\caption{Paper statements and formal theorems.}\label{tab:lean}
\end{table}

Not formalized: that the rules of Definition~\ref{def:sigma} give exactly the spiralities of rectilinear representations, and the equivalence of Lemma~\ref{lem:root} between rectilinear planarity and the test $8\in\Spir(\nu)\oplus\Band(\ell-1)$, which Lean only evaluates; both are taken from~\cite{DKLO22,DDLMO23}; Lemma~\ref{lem:parity}, which enters as the hypothesis \lean{WF}; Lemma~\ref{lem:heavy}, the formal bounds being stated in terms of $\heavyP$; the root chain of length $\ell$, which reduces to \lean{T6d_1} on the S-node whose children are $\nu$ and a Q$^*$-node of length $\ell$; the re-rooting step; the construction of $T$. The formal cost model counts unit-cost integer operations.

\section{Conclusions and Open Problems}\label{sec:concl}

We showed that rectilinear planarity of series-parallel 4-graphs, in the variable embedding setting, can be tested in $2^{O(k)}\cdot n$ time, where $k$ is the number of degree-4 vertices. In particular, if $k=o(\log n)$ the running time is $n^{1+o(1)}$, which is subquadratic and improves on the $O(n^2)$ bound of~\cite{DKLO22}. The key ingredient is the Key Lemma: light P2-nodes with a light child do not increase the number of pieces of spirality sets. This fact and the other combinatorial bounds are verified in Lean.

Our analysis actually gives a running time of $2^{O(h)}\cdot n$, where $h=\heavyP(T)\le2k$ is the number of heavy P-nodes. It is natural to ask whether the parameter can be restricted to the number of P3-nodes alone. This would follow if P2-nodes with degree-4 poles did not increase the number of pieces, as light P2-nodes do by the Key Lemma; we do not know whether this holds. Another direction is to improve the running time: the factor $36^{4k}$ comes from a crude bound on the P2- and P3-rules, and a finer analysis of how pieces merge may lead to a polynomial dependence on $k$.

Two further extensions seem within reach of known techniques, but remain to be worked out. The first concerns partial 2-trees, which need not be biconnected: as in~\cite{DKLO22}, one can combine the biconnected components along the block-cut tree, and the question is whether the bound on the number of pieces is preserved through this combination. The second is bend minimization: instead of a set of spiralities, each node stores the minimum number of bends for each spirality, as done in~\cite{DKLO22plane,DBLP:journals/corr/abs-2502-03309}, and one should check that the cost functions inherit a piecewise structure with $2^{O(k)}$ pieces.

\section*{Declaration of AI Use}
The author used ChatGPT (OpenAI) and Claude (Anthropic) as assistants in developing and checking proofs, in writing the Lean formalization, and in editing the text.

\newpage

\bibliographystyle{plainurl}
\bibliography{refs}

@article{DKLO22,
  author       = {Walter Didimo and Michael Kaufmann and Giuseppe Liotta and Giacomo Ortali},
  title        = {Rectilinear Planarity of Partial 2-Trees},
  journal      = {Journal of Graph Algorithms and Applications},
  volume       = {27},
  number       = {8},
  pages        = {679--719},
  year         = {2023},
  doi          = {10.7155/jgaa.00640},
  eprint       = {2208.12558},
  archivePrefix= {arXiv}
}

@inproceedings{DLOP20,
  author    = {Walter Didimo and Giuseppe Liotta and Giacomo Ortali and Maurizio Patrignani},
  title     = {Optimal Orthogonal Drawings of Planar 3-Graphs in Linear Time},
  booktitle = {Proceedings of the 2020 ACM-SIAM Symposium on Discrete Algorithms (SODA 2020)},
  pages     = {806--825},
  publisher = {SIAM},
  year      = {2020},
  doi       = {10.1137/1.9781611975994.49}
}

@misc{DBLP:journals/corr/abs-2502-03309,
  author       = {Walter Didimo and Giuseppe Liotta and Giacomo Ortali and Maurizio Patrignani},
  title        = {Optimal Orthogonal Drawings in Linear Time},
  year         = {2025},
  eprint       = {2502.03309},
  archivePrefix= {arXiv},
  note         = {Version 1}
}

@article{GarTam01,
  author  = {Ashim Garg and Roberto Tamassia},
  title   = {On the Computational Complexity of Upward and Rectilinear Planarity Testing},
  journal = {SIAM Journal on Computing},
  volume  = {31},
  number  = {2},
  pages   = {601--625},
  year    = {2001},
  doi     = {10.1137/S0097539794277123}
}

@inproceedings{GM00,
  author    = {Carsten Gutwenger and Petra Mutzel},
  title     = {A Linear Time Implementation of {SPQR}-Trees},
  booktitle = {Graph Drawing (GD 2000)},
  editor    = {Joe Marks},
  series    = {LNCS},
  volume    = {1984},
  pages     = {77--90},
  publisher = {Springer},
  year      = {2000},
  doi       = {10.1007/3-540-44541-2_8}
}

@article{DiBattistaLV98,
  author  = {Giuseppe {Di Battista} and Giuseppe Liotta and Francesco Vargiu},
  title   = {Spirality and Optimal Orthogonal Drawings},
  journal = {SIAM Journal on Computing},
  volume  = {27},
  number  = {6},
  pages   = {1764--1811},
  year    = {1998},
  doi     = {10.1137/S0097539794262847}
}

@misc{DDLMO23,
  author       = {Emilio {Di Giacomo} and Walter Didimo and Giuseppe Liotta and Fabrizio Montecchiani and Giacomo Ortali},
  title        = {On the Parameterized Complexity of Bend-Minimum Orthogonal Planarity},
  year         = {2023},
  eprint       = {2308.13665},
  archivePrefix= {arXiv},
  note         = {Version 3; GD 2023}
}

@misc{DKLO22plane,
  author       = {Walter Didimo and Michael Kaufmann and Giuseppe Liotta and Giacomo Ortali},
  title        = {Computing Bend-Minimum Orthogonal Drawings of Plane Series-Parallel Graphs in Linear Time},
  year         = {2022},
  eprint       = {2205.07500},
  archivePrefix= {arXiv}
}

@book{DETT99,
  author    = {Giuseppe {Di Battista} and Peter Eades and Roberto Tamassia and Ioannis G. Tollis},
  title     = {Graph Drawing: Algorithms for the Visualization of Graphs},
  publisher = {Prentice-Hall},
  year      = {1999}
}

\end{document}